\documentclass[runningheads]{llncs}

\usepackage{amsmath,amssymb,epsfig,amsfonts}
\usepackage{algorithmic,algorithm}
\newtheorem{invariant}{Invariant}

\newtheorem{observation}{Observation}
\begin{document}
\pagestyle{headings}
\titlerunning{D$^2$-Tree: A New Overlay with Deterministic Bounds}
\authorrunning{G.S.~Brodal, S.~Sioutas, K.~Tsichlas, and C.~Zaroliagis}

\newcommand{\REMOVED}[1]{}

\title{D$^2$-Tree: A New Overlay with Deterministic Bounds}

\author{Gerth St{\o}lting Brodal\inst{1}
\and
Spyros Sioutas \inst{2}
\and \\
Kostas Tsichlas \inst{3}
\and
Christos Zaroliagis \inst{4}\thanks{Part of this work was done while the author was visiting the Karlsruhe Institute of Technology.}
}

\institute{MADALGO, Department of Computer Science, University of Aarhus, Denmark, gerth@madalgo.au.dk
\and
Department of Informatics, Ionian University, Corfu Greece, sioutas@ionio.gr
\and
Department of Informatics, Aristotle University of Thessaloniki, Greece, tsichlas@csd.auth.gr
\and
Dept.~of Computer Eng.~\& Informatics, University of Patras, Greece\\
Institute of Theoretical Informatics, Karlsruhe Institute of Technology, Germany
zaro@ceid.upatras.gr}

\maketitle
\centerline{\today}

%=========================================================================
%  Abstract
%=========================================================================

\begin{abstract}
We present a new overlay, called the {\em Deterministic Decentralized tree} ($D^2$-tree).
The $D^2$-tree compares favourably to other overlays for the following reasons:
(a) it provides matching and better complexities, which are deterministic for the supported operations;
(b) the management of nodes (peers) and elements are completely decoupled from each other; and
(c) an efficient deterministic load-balancing mechanism is presented for the uniform distribution
of elements into nodes, while at the same time probabilistic optimal bounds are provided
for the congestion of operations at the nodes. The load-balancing scheme of elements into nodes
is deterministic and general enough to be applied to other hierarchical tree-based overlays.
This load-balancing mechanism
is based on an innovative lazy weight-balancing mechanism, which is interesting in its own right.
\end{abstract}
\centerline{{\bf Keywords}: Overlay, indexing scheme, decentralized system, distributed data structure, load-balancing.}
%\end{titlepage}

%=========================================================================
%  Introduction
%=========================================================================

\section{Introduction}
Decentralized systems and in particular Peer-to-Peer (P2P) networks have become very popular
of late and are widely used for sharing resources and store very large data sets. Data
are stored at the nodes (or peers) and the most crucial operations are data search
(identify the node that stores the requested information) and updates (insertions/deletions
of data). Searching and updating is typically done by building a logical \emph{overlay network}
that facilitates the assignment and indexing of data at the nodes. Sometimes, we distinguish
between the overlay structure per se and the indexing scheme used to access the data.

Following the typical modeling, a decentralized communication network
is represented by a graph. Its nodes correspond to the network nodes,
while its edges correspond to communication links.
We assume constant size messages between nodes through links and asynchronous communication.
It is assumed that the network provides an upper bound on the time needed for a node to send
a message and receive an acknowledgment. The complexity of an operation is measured in terms
of the number of messages issued during its execution. Throughout the paper, when we refer
to cost we shall mean number of messages (internal computations at nodes are considered
insignificant). The \textit{overlay} is another graph defined
over the communication network. The nodes of the overlay correspond to nodes of the original network,
while its edges (links) %%% (also referred to as links)
may not correspond to existing communication links, but to communication paths.

With respect to its \textit{structure}, the overlay supports the operations \emph{Join} (of a new node
$v$; $v$ communicates with an existing node $u$ in order to be inserted into the overlay), and
\emph{Departure} (of an existing node $u$; $u$ leaves the overlay announcing its intent to
other nodes of the overlay).

The overlay is used to implement an \emph{indexing scheme} for the stored data.
Such a scheme supports the operations \emph{search} for an element, \emph{insert}
a new element, \emph{delete} an existing element, and \emph{range query} for elements
in a specific range.

In terms of efficiency, an overlay network should address the following issues:

\begin{itemize}
\item \emph{Fast queries and updates:} updates and queries must be executed
in a minimal number of communication rounds and using a minimal number of
messages.

\item \emph{Ordered data:} keeping the data in order facilitates the implementation of various
enumeration queries when compared to a simple dictionary that can only answer membership queries,
including those arising in DNA databases, location-based services, and prefix searches for file
names or data titles. Indeed, the ever-wider use of P2P infrastructures has found applications that
require support for range queries (e.g., \cite{djxk09}).

\item \emph{Size of nodes (peers):} the size of a node is the routing information (links and related data)
maintained by this node and it is not related to the number of data elements stored in it. Keeping the size of
a node small allows for more efficient update operations, but in general reduces the efficiency of access
operations while aggravating fault tolerance.

\item \emph{Fault Tolerance:} the structure should be able to discover and heal failures at nodes or links.

\item \emph{Congestion:} it refers to the distribution of the load of search (access) operations per node, aiming at
distributing this load equally across all nodes. The congestion is an \emph{expected} quantity defined
as the maximum, among all nodes, of the fraction of the expected number of search operations at a node,
due to a random sequence of search operations on the structure, divided by the total number of
search operations.

\item \emph{Load Balancing:} it refers to the distribution of data elements on the nodes.
The goal of load balancing is to distribute equally the $n$ elements stored in the $N$
nodes of the network (typically $N\ll n$).
That is, if there are $N$ nodes and $n$ data elements,
ideally each node should carry approximately
$k$ elements, where $\lfloor n/N \rfloor \leq k \leq \lfloor n/N \rfloor +1$.

\end{itemize}

There has been considerable recent work in devising effective distributed search and update techniques.
Existing structured P2P systems can be classified into two broad categories: distributed hash table (DHT)-based
systems and tree-based systems. Examples of the former, which constitute the majority,
include Chord \cite{KKSMB01}, CAN \cite{RFHKS01}, Pastry \cite{RD01}, Symphony \cite{MBR2003}, and Tapestry \cite{ZHSRJK04}.
DHT-based systems support exact match queries well and use (successfully) probabilistic methods to
distribute the workload among nodes equally. DHT-based systems work with little synchrony and high \emph{churn}
(the collective effect created by independent burstly arrivals and departures of nodes),
a fundamental characteristic of the Internet. Since hashing destroys the ordering on keys, DHT-based
systems typically do not possess the functionality to support straightforwardly
range queries, or more complex queries based on data ordering (e.g., nearest-neighbor and string
prefix queries). Some efforts towards addressing range queries have been made in \cite{GAA2003,SGAA2004},
getting however approximate answers and also making exact searching highly inefficient.
The most recent effort towards range queries is reported in \cite{zllll09}.

Tree-based systems are based on hierarchical structures. They support range queries more
naturally and efficiently as well as a wider range of operations, since they maintain the ordering of
data. On the other hand, they lack the simplicity of DHT-based systems, and they do not always guarantee data
locality and load balancing in the whole system. Important examples of such systems include
Family Trees \cite{zh04}, BATON \cite{jov05}, and Skip List-based schemes \cite{P90}
like Skip Graphs (SG) \cite{as03,gbm04}, NoN SG \cite{mnw04}, SkipNet (SN), Deterministic SN \cite{hm03},
Bucket SG \cite{akk04}, Skip Webs \cite{aeg05}, Rainbow Skip Graphs (RSG) \cite{gns06}
and Strong RSG \cite{gns06} that use randomized techniques to create and maintain the hierarchical structure.

In this work, we focus on tree-based overlay networks that support directly range and more complex queries.
Let $N$ be the number of nodes present in the network and let $n$ denote the size of data ($N\ll n$).
Let $M$ be the size of available memory at each node, $Q(n,N)$ be the cost of a single query, $U(n,N)$ be
the cost of an update, $C(n,N)$ be the congestion per node (measuring the load) incurred by search operations,
and let $L(n,N)$ be the cost for load balancing the overlay with respect to (w.r.t.) element updates.
Regarding congestion, each node issues one operation, while the destination node of the
operation is assumed to be selected uniformly at random among all nodes of the network. Congestion depends
on the distribution of elements into nodes as well as on the topology of the overlay. It provides hints
as to how well the structure avoids the existence of \emph{hotspots} (i.e., nodes which are accessed
multiple times during a sequence of operations -- the root of a tree is usually a hotspot in decentralized tree structures).

\begin{table}[htbp]
	\begin{center}
		\begin{tabular}{| l | c | c | c | c | c | c |}
    \hline
     \scriptsize Methods &  \scriptsize $N$ &  \scriptsize $M$ &  \scriptsize $Q(n,N)$ &  \scriptsize $U(n,N)$ &  \scriptsize $C(n,N)$ &  \scriptsize $L(n,N)$\\ \hline \hline

     \scriptsize SG \cite{as03,gbm04}  &  \scriptsize $\leq n$ &  \scriptsize $O(\log{N})$ &  \scriptsize $\widehat{O}(\log{N})$ w.h.p. &  \scriptsize $\widehat{O}(\log{N})$ w.h.p. &  \scriptsize $\widehat{O}(\frac{\log{N}}{N})$ & \scriptsize $\widetilde{O}(\log{N})$ \\ \hline

     \scriptsize NoN SG \cite{mnw04} &  \scriptsize $n$ &  \scriptsize $O(\log^2{n})$ &  \scriptsize $\widehat{O}(\frac{\log{n}}{\log\log{n}})$ &  \scriptsize $\widehat{O}(\log^{2}{n})$ &  \scriptsize $\widehat{O}(\frac{\log^2{n}}{n})$ & -- \\ \hline

\scriptsize Determ.~SN \cite{hm03} &  \scriptsize $n$ &  \scriptsize $O(\log{n})$ &  \scriptsize $O(\log{n})$ &  \scriptsize $O(\log^2{n})$ &  \scriptsize $O(\frac{n^{0,32}}{n})$ & -- \\ \hline

     \scriptsize BATON \cite{jov05} &  \scriptsize $\leq n$ &  \scriptsize $O(\log{N})$ &  \scriptsize $O(\log{N})$ &  \scriptsize $O(\log{N})$ & -- &  \scriptsize $\overline{O}(\log{n})$\\ \hline

%    m-ary BATON \cite{jotvz06} (nodes) && $O(m\log_m{n})$ & $O(m\log_m{n})$ & $O(m\log_m{n})$ & - & \\ \hline
     \scriptsize Family Trees \cite{zh04} &  \scriptsize $n$ &  \scriptsize $O(1)$ &  \scriptsize $\widehat{O}(\log{n})$ &  \scriptsize $\widehat{O}(\log{n})$ &  \scriptsize $\widehat{O}(\frac{\log{n}}{n})$ & -- \\ \hline

     \scriptsize Bucket SG \cite{akk04} &  \scriptsize $\leq n$ &  \scriptsize $O(\frac{n}{N}+\log{N})$ &  \scriptsize $\widehat{O}(\log{N})$ &  \scriptsize $\widehat{O}(\log{N})$ &  \scriptsize $\widehat{O}(\frac{1}{N}+\frac{\log{N}}{n})$ &  \scriptsize No Bounds \\ \hline

     \scriptsize Skip Webs \cite{aeg05} &  \scriptsize $n$ &  \scriptsize $O(\log{n})$ &  \scriptsize $\widehat{O}(\frac{\log{n}}{\log\log{n}})$ &  \scriptsize $\widehat{O}(\frac{\log{n}}{\log\log{n}})$ &  \scriptsize $\widehat{O}(\frac{\log{n}}{n})$ & -- \\ \hline

     \scriptsize Rainbow SG \cite{gns06} &  \scriptsize $n$ &  \scriptsize $O(1)$ &  \scriptsize $\widehat{O}(\log{n})$ w.h.p. &  \scriptsize $\overline{O}(\log{n})$ w.h.p. &  \scriptsize $\widehat{O}(\frac{\log{n}}{n})$ & -- \\ \hline

     \scriptsize Strong RSG \cite{gns06} &  \scriptsize $n$ &  \scriptsize $O(1)$ &  \scriptsize $O(\log{n})$ & \scriptsize $\widetilde{O}(\log{n})$ &  \scriptsize $\widehat{O}(\frac{n^\epsilon}{n})$ & -- \\ \hline \hline

     \scriptsize \textbf{D$^2$-tree} &  \scriptsize $\leq n$ &  \scriptsize $O(1)$ &  \scriptsize $O(\log{N})$ & \scriptsize $\widetilde{O}(\log{N})$ &  \scriptsize $\widehat{O}(\frac{\log{N}}{N})$ &  \scriptsize $\widetilde{O}(\log{N})$ \\ \hline
\hline
  \end{tabular}
	\end{center}
\caption{A comparison between previous methods and the $D^2$-tree. By $\widehat{O}$ we represent expected bounds, by $\widetilde{O}$ we represent amortized bounds,
and by $\overline{O}$ expected amortized bounds. All other bounds are worst-case. Typically, $N\ll n$.}
	\label{tab:comp}
\end{table}
A comparison of the aforementioned tree-based overlays is given in Table~\ref{tab:comp}.
We would like to emphasize that w.r.t.~load balancing,
there are solutions in the literature either as part of the overlay (e.g., \cite{jov05})
or as a separate technique (e.g., \cite{akk04,gbm04}). These solutions are either heuristics, or
provide expected bounds under certain assumptions, or amortized bounds but at the expense of increasing
the memory size per node. In particular, in BATON \cite{jov05}, a decentralized overlay is provided with load balancing based on
data migration. However, their $O(\log{n})$ amortized bound is valid only subject to a
probabilistic assumption about the number of nodes taking part in the data migration process,
and thus it is in fact an amortized expected bound. In the case of Bucket Skip Graphs \cite{akk04},
elements are structured in buckets attached to nodes. Although it is a solution
which can be applied to a large set of P2P structures, it has two drawbacks: (i) a list of
free nodes is required, and (ii) a global control for the size of the buckets is imperative.
The latter is very crucial and is tackled by heuristics with no analysis whatsoever.
The solution proposed in \emph{this} paper can be used to tackle the problem
of bucket size control efficiently in an amortized sense.
A deterministic solution for load-balancing comes from \cite{gbm04}, in which a $O(\log{N})$ amortized bound
w.r.t. the elements transferred is provided. Their solution,
stemming from a centralized parallel database framework,
is a node migration process in which a lightweight node is selected,
its load is moved to an adjacent node and then it shares the load of the heavyweight node.
This process was initially developed for a parallel database in which there is central control.
The original process was translated to a decentralized framework by applying a second overlay on the nodes
where the order is defined w.r.t.~the load of the nodes. In particular, they maintain two
skip graphs on the nodes, one w.r.t.~the order of elements and one w.r.t.~the load
of the nodes (in fact the second one can be replaced by a decentralized min-heap \cite{ss09}).
Apart from this deficit, one more problem with this method is that it assumes that node migration
is possible and each time an update takes place the structure of the overlay is changed. This incurs an additive cost equal to the cost update of the
structure. Additionally, in structures that strive for deterministic bounds (like BATON)
this is not possible since such structures are quite strict and do not allow the placement of a node anywhere in the structure.

The basic characteristic of a decentralized overlay is that the balancing information is local.
Locality is a must in a decentralized structure since there are no means to acquire global information.
For example, internal memory height-balanced trees have local balancing information and thus lend themselves
nicely to P2P environments but they have problems with congestion w.r.t. updates. In particular,
in a sequence of $n$ operations the root can be accessed $O(\sqrt{n})$ times. However, weight balanced trees
avoid this bottleneck having very good congestion w.r.t. updates but they need a lazy mechanism
as the one described in this paper to update the weight information.

\paragraph{\textbf{Our Contribution.}}
In this paper we present a new tree-based overlay, called the
\emph{Deterministic Decentralized tree} or \textit{$D^2$-tree}. The
$D^2$-tree (see also Table~\ref{tab:comp}):
\begin{itemize}
\item uses $O(1)$ space per node;

\item achieves a deterministic $O(\log N)$ query bound;

\item achieves a deterministic (amortized) $O(\log{N})$ update bound for elements as well as for node joins
and departures;

\item achieves \emph{optimal} congestion;

\item exhibits a deterministic (amortized) $O(\log{N})$ bound for load-balancing;

\item supports ordered data queries optimally, and tolerates node failures.
\end{itemize}

The $D^{2}$-tree is an overlay consisting of two levels. The upper level is a perfect binary tree. The leaves of this tree are representatives of the buckets that constitute the lower level of the $D^2$-tree. Each bucket is a set of nodes and these nodes are structured as a doubly linked list. Each bucket contains $O(\log{N})$ nodes. Since $N$ changes, the size of buckets is dynamically maintained by the overlay.

In the $D^2$-tree, we separate the index from the overlay structure using the
load-balancing mechanism. The number of elements per node is dynamic w.r.t.
 node joins and departures and it is controlled by the load-balancing mechanism.
Moreover, the number of nodes of the perfect binary tree is not connected by any means
to the number of elements stored in the structure.
The overlay structure supports the operations of node join and node departure, while at
the same time it tackles failures of nodes whenever these are discovered.
%%% The overlay index supports all the
%%%% aforementioned operations on the stored data (search, update and range queries).

Our load-balancing technique distributes almost equally the elements among nodes
by making use of weights. Weights are used to define a metric of load-balance,
which shows how uneven is the load between nodes. When the load is uneven, then
a data migration process is initiated to equally distribute elements.

Our load-balancing technique is quite general and can be applied to any hierarchical
decentralized overlay (e.g., BATON, Skip Graphs) with the following specifications:
\begin{itemize}
\item The overlay structure must be a tree with height $O(\log{N})$ and with each node having $O(1)$ children.

\item Nodes at level $i$ having the same father have approximately (within constant factors) the same weight, which is $\Omega(i^4)$.

\item Updates are performed at the leaves. Alternatively, if each node has access to a leaf in $O(1)$ messages
then this is enough, since the update is simply forwarded to this leaf.
\end{itemize}

The rest of the paper is organized as follows. Section~\ref{sec:def} presents
some definitions and notation used throughout the paper. We discuss the load balancing technique
in Section~\ref{sec:load}, and present the $D^{2}$-tree in Section~\ref{sec:d2}.
We conclude in Section~\ref{sec:conclusion}. A preliminary version of this work
appeared as \cite{BSTZ2010}.

\section{Definitions and Notation}
\label{sec:def}

In this section, we give some definitions regarding tree structures that will
be used throughout the paper.

Let $\mathcal{T}$ be a tree.
Based on $\mathcal{T}$ ancestor-descendant relationships are defined in a natural way. There is a node that
has no ancestor (the \textit{root}) and there are nodes with no descendants (the \textit{leaves}).
All nodes which are not leaves are called \textit{internal}. The subgraph induced by the descendants
of node $v$ (including $v$) in $\mathcal{T}$ is the \textit{subtree} of $v$.
The \textit{height} of node $v$ is the length (in number of edges) of the longest path from $v$ to one of its leaves.
The \textit{depth} or \textit{level} of node $v$ is the length of the path from $v$ to the root.
Two nodes are called \textit{brothers} when they have the same father and they are consecutive in his child list.

The \textit{weight} $w(v)$ of a node $v$ is equal to the number of elements stored in its subtree.
%%% The term weight will also be used to express other similar quantities at some parts of the paper,
%%% in which case we explicitly say so.
The number of elements residing in a node $v$ is denoted by $e(v)$.
%%% A relaxed version of the weight of $v$ is the so-called {\em virtual weight} $b(v)$
%%% of $v$ defined as the weight stored in node $v$.

We define the \emph{size} of $v$, denoted by $|v|$, as the number
of nodes of the subtree of $v$ (including $v$) in $\mathcal{T}$.
The \emph{density} $d(v)$ of node $v$ is defined as $d(v)=\frac{w(v)}{|v|}$ and represents
the mean number of elements per node in the subtree of $v$.

Let $v$ be a node at height $h$, let $p$ be a child of $v$ and let $q$ be the right
brother of $p$; both $p$ and $q$ are at height $h-1$.

The \textit{criticality} $c(p,q)$ of the two brother nodes $p$ and $q$
is defined as $c(p,q)=\frac{d(p)}{d(q)}$ and represents their difference in densities.

Let $\mathcal{T'}$ be a perfect binary tree.
The \textit{node criticality} $nc_v$ of a node $v\in\mathcal{T'}$ at level $\ell$ with left and right
children $w$ and $z$ at level $\ell+1$, respectively, is defined as $nc_v=\frac{|w|}{|v|}$.
The node criticality represents the difference in size between a node ($v$) and its
left child ($w$).

\section{Deterministic Load Balancing}
\label{sec:load}

The main idea of our load-balancing mechanism is as follows.
It distributes almost equally the elements among nodes by making use of weights,
which are used to define a metric showing how uneven is the load between nodes.
When the load is uneven, then a data migration process is initiated to equally
distribute the elements.

We describe the load-balancing mechanism in two steps. First, we provide a mechanism that allows
for efficient and local update of weight information in a tree when elements are added or removed
at the leaves. This is necessary to avoid hotspots. Then, we describe the load-balancing
scheme in a tree overlay. In the following, we assume that the overlay structure is a tree $\mathcal{T}$.

\subsection{A Technique for Amortized Constant Weight Updating}
\label{ssec:rebalance}

We provide a technique that lazily updates the weights on the nodes of a tree. When an element is
added/removed to/from a leaf $u$ in $\mathcal{T}$, the weights on the path from $u$ to the root
must be updated. If the height of $\mathcal{T}$ is $H$, then the cost of the weight updating is $O(H)$.
Assume that node $v$ lies at height $h$ and its children are $v_{1},v_{2},\ldots,v_{s}$ at height $h-1$.
We relax the weight of a node and its recomputation. We define the {\em virtual weight} $b(v)$ of $v$ as the
weight stored in node $v$. In particular, for node $v$ the following invariants are maintained

\begin{invariant} \label{inv:insertions}
$b(v)>e(v)+(1-\epsilon_{h})\left(\sum_{i=1}^{s}{b(v_{i})}\right)$
\end{invariant}
\begin{invariant} \label{inv:deletions}
$b(v)<e(v)+(1+\epsilon'_{h})\left(\sum_{i=1}^{s}{b(v_{i})}\right)$
\end{invariant}
where $\epsilon_{h}$ and $\epsilon'_{h}$ are appropriate constants. These invariants imply that the weight
information is approximate, at most by a multiplicative constant.

Assume that an update takes place at leaf $u$. Apparently, only the weight of its ancestors need to be updated by $\pm 1$ and no other node is affected. We traverse the path from $u$ to the root until we find a node $z$ for which Invariants~\ref{inv:insertions} and \ref{inv:deletions} hold. Let $v$ be its child for which either Invariant~\ref{inv:insertions} or \ref{inv:deletions} does not hold on this path. We recompute all weights on the path from $u$ to $v$. In particular, for each node $z$ on this path, we update its weight information by taking the sum of the weights written in its children plus the number of elements that $z$ carries.

The constants $\epsilon_{h}$ and $\epsilon'_{h}$ are chosen such that for all nodes
the virtual weight will be within a constant factor $c>0$ of the real weight, i.e.,
$$\frac{1}{c}\cdot w(v) < b(v) < c\cdot w(v)$$
First we prove the lower bound on $v$. At height $h$:
$$b(v) > (1-\epsilon_{h})\left(\sum_{j=1}^{s}{b(v_{j})}+e(v)\right)$$
By recursing and lower bounding to get clean bounds we get
\[b(v) > w(v)\prod_{j=2}^{h}{(1-\epsilon_{j})}\]
Choosing\footnote{We have chosen this $\epsilon_{j}$ for simplicity.
In fact for any $\eta>0$, choosing $\epsilon_{j}=\frac{1}{j^{1+\eta}}$ is sufficient.}
 $\epsilon_{j}=\frac{1}{j^{2}}$ , we get
\[\prod_{j=2}^{h}{\left(1-\frac{1}{j^{2}}\right)}=\frac{\prod_{j=2}^{h}{(j-1)}\times\prod_{j=2}^{h}{(j+1)}}{\prod_{j=2}^{h}{j^2}}
= \frac{(h-1)!\prod_{j=3}^{h}{j}}{(h!)^2}=\frac{h+1}{2h}>\frac{1}{2}\]
Similarly, for the upper bound we get
$b(v) < w(v)\prod_{j=1}^{h}{(1+\epsilon'_{j})}$.
Choosing $\epsilon_{j}=\frac{1}{j^{2}}$ and taking into account that $1+\frac{1}{j^{2}}<\frac{1}{1-\frac{1}{j^2}}$, we have
\[\prod_{j=2}^{h}{\left(1+\frac{1}{j^{2}}\right)} < \prod_{j=2}^{h}{\left(\frac{1}{1-\frac{1}{j^2}}\right)}=
\frac{1}{\prod_{j=2}^{h}{\left({1-\frac{1}{j^2}}\right)}}<\frac{1}{\frac{1}{2}} = 2\]
As a result, by choosing $\epsilon_{h}=\epsilon'_{h}=\frac{1}{h^{2}}$ we get that:
\begin{equation}
\label{eq:vw}
\frac{1}{2}\cdot w(v) < b(v) < 2 \cdot w(v)
\end{equation}

The following lemma states how frequently the weight information in each node changes.
%%% For example, in a sequence of $n$ operations, the weight at the root will be updated $O\left(\frac{H^{1+\eta}}{n}\right)$
%%% times, for any $\eta>0$.

\begin{lemma} \label{lem:rebalancing}
The minimum number of updates in the subtree of $v$, causing a weight update at $v$, is $\Theta(\epsilon_{h}w(v))$.
\end{lemma}
\begin{proof}
The weight update of node $v$ is a result of the violation of either of Invariants~\ref{inv:insertions} or \ref{inv:deletions}.  After the update, it holds that $b(v)=\sum_{i=1}^{s}{b(v_i)}+e(v)$. Node $v$ has its weight updated again when $b(v)<(1-\epsilon_{h})\left(\sum_{i=1}^{s}b(v_i)+e(v)\right)$ (or $b(v)>(1+\epsilon_{h})\left(\sum_{i=1}^{s}{b(v_i)}+e(v)\right)$ symmetrically). This will happen only when the weight of the subtree of $v$ changes by $\epsilon_{h}\left(\sum_{i=1}^{s}{b(v_i)}+e(v)\right)$. This change is a lower bound on the number of operations performed in this subtree,
no matter when they have been performed.
Taking into account (\ref{eq:vw}), we get the lemma.
\qed
\end{proof}
The following theorem states that the weight updating mechanism is efficient in an amortized sense.
\begin{theorem} \label{thm:wupdate}
The amortized cost of the weight update algorithm is $O(1)$.
\end{theorem}
\begin{proof}
Lemma~\ref{lem:rebalancing} states that if we make $\epsilon_{h} w(v)$ update operations then the maximum number of weight changes
at node $v$ is $1$. As a result, the amortized cost per update operation at height $h$ is $\frac{1}{\epsilon_{h}b(v)}$.
In the following, given that $v^{(i)}$ is the node on the path at height $i$ and by the assumption that $b(v^{(i)})=\Omega(i^4)$ we get that the amortized cost is:
\[\sum_{i=0}^{H}{\frac{1}{\epsilon_{i}b(v^{(i)})}}=\sum_{i=0}^{H}{\frac{i^2}{b(v^{(i)})}}=\sum_{i=0}^{H}{O\left(\frac{i^2}{\Omega(i^4)}\right)}=O(1) \mbox{\hfill \qed} \]
\end{proof}

\subsection{Updates and Load Balancing}
\label{ssec:load}

We now investigate how load balancing is realized on the balanced tree structure $\mathcal{T}$. For clarity of exposition,
we assume that $\mathcal{T}$ is a binary tree. The following discussion can be easily generalized for trees with $O(1)$
maximum degree, simply by looking between brother nodes.

First, bear in mind that this mechanism does not tamper with the structure of $\mathcal{T}$. An update operation
(either insertion or deletion of an element) is initiated at node $v$. Node $v$ issues a search for the involved
element and the appropriate node $u$ is returned. Then, the update request is forwarded from $v$ to $u$. Node $u$
executes the update operation and signals $v$ for the status of the update. The load balancing mechanism redistributes
the elements among nodes when the load between nodes is not distributed equally enough.

Assume that node $v$ at height $h$ has child $p$ and its right brother $q$ at height $h-1$.
Recall that $|v|$ denotes the \emph{size} of $v$ (number of nodes in the subtree of $v$, including $v$)
in the overlay structure, $d(v)=\frac{w(v)}{|v|}$ denotes the \emph{density} of $v$ (representing
the mean number of elements per node in the subtree of $v$), and that $c(p,q)=\frac{d(p)}{d(q)}$
denotes the \textit{criticality} of the two brother nodes $p$ and $q$ (representing their difference in densities).
The following invariant guarantees that there will not be large differences between densities.
\begin{invariant} \label{inv:dens}
For two brothers $p$ and $q$, it holds that $\frac{1}{c} \leq c(p,q) \leq c$, $1< c \leq 2$.
\end{invariant}
For example, choosing $c=2$ we get that the density of any node can be at most twice or half of that of its brother.
In the more general case where the number of children of node $v$ is $O(1)$, we get that no child of $v$ has more
density than a constant factor w.r.t. the other children of $v$.

When an update takes place at leaf $u$, weights are updated by using the mechanism described in Section~\ref{ssec:rebalance}.
In this way, we guarantee that no hotspot exists w.r.t. weight updating as implied by Lemma~\ref{lem:rebalancing}.
Then, starting from $u$, the highest ancestor $w$ is located that is unbalanced w.r.t. his brother $z$,
meaning that Invariant~\ref{inv:dens} is violated. Finally,
the elements in the subtree of their father $v$ are redistributed uniformly so that the density of the brothers becomes equal;
this procedure is henceforth called \textit{redistribution} of node $v$. Assume that the redistribution phase
has a cost of $O(f(w(v)))$, for some increasing function $f:\mathbb{N} \rightarrow  \mathbb{N}$.
The following theorem provides amortized bounds for the redistribution.
\begin{theorem}
\label{thm:redistcost}
 The load balancing has an amortized cost of $O\left(H\frac{f(n)}{n}\right)$.
\end{theorem}
\begin{proof}
If a node $v$ with weight $w(v)$ has the elements in its subtree redistributed, then this node will go
through this process again after $O(w(v))$ updates of elements in its subtree. In particular, when $v$
is redistributed the criticality $c(p,q)$ of its children $p$, $q$ is $1$. To move the criticality out
of bounds again at least $\frac{w(p)}{2}$ or $\frac{w(q)}{2}$ elements must be inserted or deleted from
$p$ or $q$ respectively. By the assumption that the number of nodes in the subtree of $p$ is approximately
equal (within constant factors) to that of $q$, we deduce that $O(w(v))$ elements must be inserted or
deleted from $v$. Since the cost of the redistribution of $v$ is $O(f(w(v)))$, the amortized cost
for node $v$ is $O\left(\frac{f(w(v))}{w(v)}\right)$. This is true for all nodes on the path from
a leaf to the root, and thus the amortized cost is $O\left(H\frac{f(w(root))}{w(root)}\right)$.
\qed
\end{proof}

\section{The $D^2$-tree} \label{sec:d2}
In this section we design and analyze the {\em $D^2$-tree} overlay.
We first describe the overlay structure, then move to the description of the index, and
finally discuss efficiency issues regarding congestion and fault-tolerance.

\subsection{The $D^2$-tree Structure} \label{ssec:overlay}

The $D^2$-tree is a binary tree, where each node maintains an additional set of links to other
nodes apart from the standard links which form the tree. Each node $v$ in the tree maintains the following links:
\begin{enumerate}

 \item Links to its father (if there is one) and its children.

 \item Links to its adjacent nodes based on an inorder traversal of the tree.

 \item Links to nodes at the same level as $v$. These links facilitate an exponential search on the nodes of the same level. Assume that node $v$ lies at level $\ell$. In a binary tree, the maximum number of nodes at level $\ell$ is equal to $2^{\ell}$. Node $v$ maintains at most $2\ell$ links: $\ell$ links to nodes to the right and $\ell$ links to nodes to the left. The links are distributed in exponential steps, that is the first link points to a node (if there is one) $2^{0}$ positions to the left (right), the second $2^{1}$ positions to the left (right), and the $i$-th link $2^{i-1}$ positions to the left (right). These links constitute the \textit{routing table} of $v$.
\end{enumerate}
The next lemma captures some important properties of the routing tables w.r.t. their construction.
It follows immediately from the aforementioned link structure and the fixed distances between successive links in the routing tables.
\begin{lemma}\label{lem:properties}
(i) If a node $v$ contains a link to node $u$ in its routing table, then the parent of $v$ also contains a link to
the parent of $u$, unless $u$ and $v$ have the same father.
(ii) If a node $v$ contains a link to node $u$ in its routing table, then the left (right) sibling of $v$ also contains
a link to the left (right) sibling of $u$, unless there are no such nodes.
(iii) Every non-leaf node has two adjacent nodes in the inorder traversal, which are leaves.
\end{lemma}

\subsubsection{A Weight-Balanced Overlay.} \label{ssec:weighted}

The overlay consists of two levels. The upper level of the overlay is a Perfect Binary Tree (PBT). The leaves of the tree are representatives of buckets that constitute the lower level of the overlay. Each \textit{bucket} is a set of $O(\log{N})$ nodes and it is structured as a doubly linked list. Each node of the bucket points to the node which is a leaf of the PBT and is called the \textit{representative} of the bucket. Additionally, it maintains its routing table w.r.t. the nodes of all buckets.

When a node $z$ makes a join request to $v$, then this node is forwarded to its adjacent leaf $u$ w.r.t.~the inorder traversal.
Then, node $z$ is added to the doubly linked list representing the bucket of $u$ by manipulating a constant number of links. The routing table of $z$ is updated by using Lemma~\ref{lem:properties}(ii). When a node $v$ leaves the network, then it is replaced by its right adjacent node $u$ (if there is no right adjacent node then we choose the left one) which in turn is replaced by its first node $z$ in its bucket (Figure~\ref{fig:joinleave}). Link and data information are copied from $v$ to $u$ and from $u$ to $z$. When a node $v$ is discovered to be unreachable, its adjacent node $u$ is first located. This is accomplished by traversing the path to the rightmost or leftmost leaf starting from the left or right child respectively. Node $u$ fills the gap of $v$ and the first child $z$ in the bucket of $u$ fills the gap left by $u$. The contents of $u$ are not moved to another node except from the navigation data (routing tables and other links) which are moved to node $z$ that takes its place. Node $u$ has its routing tables recomputed.

\begin{figure}[ht]
	\centering
	\includegraphics[scale=0.6]{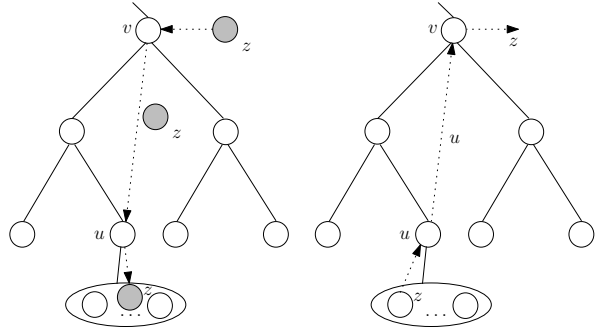}
  \caption{To the left (right) the join of $z$ (leave of $v$) is depicted. The dotted labeled arrows represent the movement of the nodes denoted by the label.}
  \label{fig:joinleave}
\end{figure}

The join and departure of nodes may cause the size of the buckets to be uneven, which in the long run
renders the structure unbalanced (imagine a bucket holding almost all nodes). To control the size of
the buckets we employ a weight-based approach\footnote{The alternative of following a height-based
approach, resulting in a height (instead of weight) balanced overlay, would render
update operations inefficient.}. Each node $v$ of the PBT maintains its size $|v|$,
which is equal to the number of nodes in the buckets of its subtree. The size control is accomplished
by using the method introduced in Section~\ref{ssec:rebalance}, in order to avoid the existence of hotspots.

Recall that the \textit{node criticality} $nc_v$ of a node $v$ at level $\ell$ with left and right children $w$ and $z$
at level $\ell+1$, respectively, is defined as $nc_v=\frac{|w|}{|v|}$.
The following invariant bounds the criticality of nodes.
\begin{invariant} \label{inv:criticality}
 The node criticality of all nodes is in the range $\left[\frac{1}{4},\frac{3}{4}\right]$.
\end{invariant}
Invariant~\ref{inv:criticality} implies that the number of nodes in buckets in the left
subtree of a node $v$ is at least half and at most twice the corresponding number of its
right subtree (this definition can be easily generalized when $v$ has a $O(1)$ number
of children).

When an update takes place at bucket $x$, then we locate the highest ancestor $v$ of $x$
whose node criticality is out of bounds, w.r.t. Invariant~\ref{inv:criticality}, and we
redistribute the nodes in its subtree.

%%% \ref{ss:rednodes}
%%% because we want to emphasize that as in the case of the load-balancing one may come up with a better implementation-wise technique to realize it.

The redistribution is carried out as follows. A traversal of all buckets of the subtree of $v$ at level $\ell$ is performed in order to determine the exact value of $|v|$. Then, the number of nodes per bucket should be $\left\lfloor \frac{|v|}{2^{\ell}}\right\rfloor+1$. The redistribution of nodes in the subtree of $v$ starts from the rightmost bucket and it is performed in an inorder fashion so that elements in the nodes are not affected. The transfer of nodes is accomplished by maintaining a link (called {\em dest} henceforth) for the position in which nodes should be put or taken from. In addition, this pointer plays the role of a token indicating which node implements the redistribution process. The transfer process involving bucket $b$ is implemented by its representative that maintains the pointer \emph{dest}.

 Assume that bucket $b$ has $q$ extra nodes which must be transferred to other buckets. Pointer \emph{dest} points to a bucket $b'$ in which these extra nodes should be put. All these nodes are put in $b'$ as well as in adjacent nodes if necessary. Note that during this procedure internal nodes of PBT are also updated since \emph{dest} implements an inorder traversal following the respective pointers. When bucket $b$ has the correct size, the link \emph{dest} is transferred to the representative of the next bucket and the same procedure applies again. In each visited bucket there are nodes which have been transferred and are in their correct position and there are nodes which are to be transferred. The distinction between these nodes is quite easy by the total number of nodes in the bucket as well as by the keys they contain. The case where $q$ nodes must be transferred to bucket $b$ from bucket $b'$ is completely symmetric. The cost for the redistribution for node $v$ is $f(|v|)=O(|v|)$.

The redistribution guarantees that if there are $z$ nodes in total in the $y$ buckets of the subtree of $v$, then after the redistribution each bucket maintains either $\lfloor z/y \rfloor$ or $\lfloor z/y \rfloor +1$ nodes. However, the following discussion still holds (with minor changes) even if the redistribution phase guarantees that the minimum and maximum size of the buckets is within constant factors. The cost for the redistribution we propose for node $v$ is $f(|v|)=O(|v|)$.

We guarantee that each bucket contains $O(\log{N})$ nodes, throughout joins or departures of nodes,
by employing two operations on the PBT, the \textit{contraction} and the \textit{extension}.
When a redistribution takes place at the root of the PBT, we also check whether any of these two operations can be applied to the PBT. The extension operation adds one more level of nodes at the PBT from existing nodes in the buckets, thus increasing its height by one. The contraction operation removes one level of nodes from the PBT and puts them into the buckets, thus decreasing its height by one. In order to decide whether the PBT needs extension or contraction we compare the size of the buckets $B$ after the redistribution with the height of the PBT. Note that after redistribution, the sizes of all buckets may differ by at most $1$. If the size is larger than the height of the PBT by at least $1$ then an extension takes place. If the size of the bucket is smaller than the height of the PBT by at least $1$ then a contraction takes place (see Figure~\ref{fig:pbtops}). The height of the PBT can be deduced by the size of the routing table in the nodes of the last level of the PBT. These two operations involve a reconstruction of the overlay which rarely happens as shown in the following lemma.

\begin{figure}[ht]
	\centering
	\includegraphics[scale=0.6]{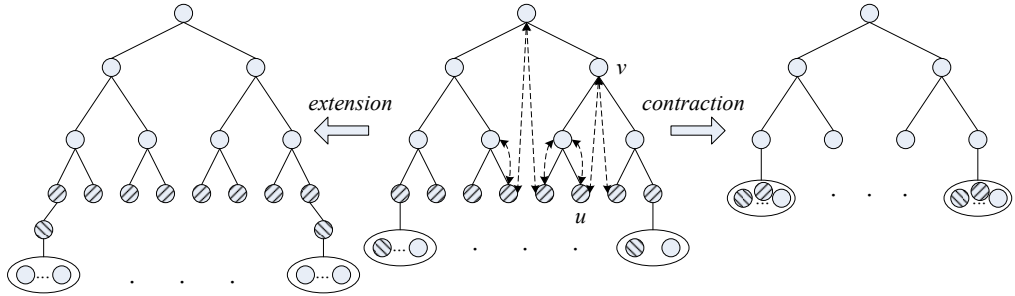}
  \caption{In the middle, the structure of the weight balanced overlay is depicted. To the left (right) is the result of the application of an extension (contraction) operation.}
  \label{fig:pbtops}
\end{figure}

\begin{lemma} \label{lem:weight}
If a redisribution operation is performed at a node with size $s$, then this node will be redistributed again
after $\Omega(s)$ joins or departures have been performed in its subtree.
\end{lemma}
\begin{proof}
Assume that node $v$ with size $s$ is redistributed. Then, $nc_v=0.5$, meaning that the number of nodes in the buckets for both subtrees are equal. The bound of $0.5$ on criticality after redistribution is not strict in the sense that any bound in the interval $[\frac{1}{4}+\zeta, \frac{3}{4}-\zeta]$, where $\zeta>0$, suffices. The same holds for their subtrees recursively. Node $v$ will be redistributed again only when the criticality of one of its children gets out of bounds. Since it was $0.5$ at least $s/4$ joins or departures of nodes must be performed in order to redistribute $v$. This is a worst-case sequence of operations that trigger a redistribution at $v$. Assuming a uniform distribution of updates, a much larger bound can be obtained. %This is a very good point to proce expected bounds even whp reducing in fact considerably the number of redistributions%
\qed
\end{proof}
Lemma~\ref{lem:weight} states that the expensive operations of extension and contraction take place when the number of nodes has at least doubled or halved.
Assuming that the redistribution of $v$ has $O(f(|v|))$ cost, it follows by Lemma~\ref{lem:weight} that the amortized cost for join/departure of a node $v$
at height $h$ is $O\left(\frac{f(|v|)}{|v|}\right)$. Since the PBT has height $H$, we establish the following.

\begin{lemma} \label{lem:migrate}
The amortized cost of join/departure of a node $v$ is $O\left(H\frac{f(N)}{N}\right)$.
\end{lemma}

\subsubsection{$O(1)$ Space per Node.}
The routing tables require $O(\log{N})$ space for each node. To make the space consumption constant, one could
apply on the overlay the schemes described in \cite{gns06,zh04}. However, on the one hand the complexities will
not be deterministic while on the other hand even in the case of Strong Rainbow Skip Graphs \cite{gns06} with
deterministic bounds our congestion for searching is much better than theirs. To achieve constant space we
distribute the routing tables to many nodes doing the same also for nodes in the buckets.
A set of nodes with constant degree is grouped together and a routing table is distributed on all these nodes,
such that each node uses constant space. Thus, a node can recreate approximately its routing table by
accessing nodes inside the same group. We call each such group a \textit{hypernode}.

A hypernode at level $\ell$ consists of at most $\ell$ nodes, numbered from left to right $1,2,\ldots$.
This number is the {\em rank} of the node within the hypernode. A node $v$ with rank $i$ maintains two links
to the nodes that are approximately $2^i$ positions to the right and to the left. In particular, node $v$ either points
to a node $z$ in the same hypernode whose distance is $2^i$ or to a node $z'$ whose rank is $i$ and lies in a different
hypernode than that of $v$ which contains a node whose distance is $2^i$ from $v$.
The concatenation of all such links constitutes the routing table for the hypernode. Additionally, each
node with rank $i$ maintains two links to nodes with ranks $i-1$ and $i+1$, if there are such nodes.
Finally, each node with rank $i$ in the hypernode maintains a link to the node with the largest rank.
The following lemma translates Lemma~\ref{lem:properties}(ii) in the setting of hypernodes.
\begin{lemma}\label{lem:hyper_routing_con_sibling}
If node $v$ contains a link to node $u$, then the left (right) sibling of $v$ also contains a link
to the left (right) sibling of $u$, unless $\nexists$ such nodes.
\end{lemma}

\begin{proof}
Direct implication of the distances between successive links in the routing tables as well as of the
increasing ranks in the hypernodes.
\qed
\end{proof}

Using Lemma~\ref{lem:hyper_routing_con_sibling} we can update the links of a node $v$ by simply looking at the links of its siblings $u$ and $w$ and update the links of $v$ by pointing to the adjacent nodes of the nodes pointed to by $u$ and $w$. Hypernodes are static in the overlay and only in the case of contraction we destroy the hypernodes of the last level while in the case of extension we create new hypernodes for the new level. A faulty node inside a hypernode will not disconnect it since by accessing the parents we can find its siblings and reconstruct the missing routing information.

\subsection{The Index Structure of the $D^2$-tree}
\label{ssec:index}

The overlay provides the infrastructure for the index to efficiently support various operations. The overlay is used as a node-oriented tree. The range of all values stored in the overlay is partitioned into subranges each one of which is assigned to a node of the overlay. An internal node $v$ with range $[x_v,x'_v]$ may have a left child $u$ and a right child $w$ with ranges $[x_u,x'_u]$ and $[x_w,x'_w]$ respectively such that $x_u<x_u'<x_v<x_v'<x_w<x_w'$. Thus, if an element $x \in [x_v,x'_v]$ then it must be stored at node $v$. Ranges are dynamic in the sense that they depend on the values maintained by the node.

 In the following, we discuss the search and update operations supported by the index. Our arguments refer
to the case where nodes use $O(\log{N})$ space but they can be trivially changed to hold in the case they
use $O(1)$ space. In the few cases where these arguments do not transfer trivially we make further explanations.

\subsubsection{Search and Range Queries.}
\label{ssec:congestion}
The search for an element $\alpha$ in the overlay may be initiated from any node $v$ at level $\ell$.
Let $z$ be the node with range of values containing $\alpha$. Assume $O(\log{N})$ space per node
and assume that w.l.o.g.~$x'_v<\alpha$. Then, by using the routing tables we search at level $\ell$
for a node $u$ with right sibling $w$ (if there is such sibling) such that $x'_u<\alpha$ and $x_w>\alpha$ unless
$\alpha$ is in the range of $u$ and the search terminates. This step has $O(\ell)$ cost, since we simulate a binary
search. If the search continues, then node $z$ will either be an ancestor of $u$ or in the subtree rooted at $u$.
If $u$ is a leaf, then we move upwards (or in its corresponding bucket) until we find node $z$ in $O(\log{N})$ steps. If $u$ is an internal node,
by following the respective link we move to the left adjacent node $y$ of $u$ which is certainly a leaf (inorder
traversal). If $x'_y>\alpha$ then an ordinary top down search from node $u$ will suffice to find $z$ in $O(\log{N})$
steps (or in its bucket). Otherwise, node $z$ is certainly an ancestor of $u$ and thus we can move upwards from $u$ until we find
it in $O(\log{N})$ steps. The following lemma establishes the complexity of the search operation.

\begin{lemma} \label{lem:search}
The search for an element $\alpha$ in a $D^2$-tree of $N$ nodes is carried out in $O(\log{N})$ steps.
\end{lemma}

\begin{proof}
The case of $O(\log{N})$ space per node was analyzed in the paragraph preceding the statement of the Lemma.

In the case of $O(1)$ space per node, assume that node $v$ belongs in the hypernode $V$ at level $\ell$.
The only change concerns the discovery of node $u$ at level $\ell$. By following the respective link, node $p \in V$
with highest rank is reached. Then, by following the backward links we make the search on the level $\ell$. In particular,
assume that during our search in hypernode $V$ we find that node $u$ is somewhere between the nodes pointed to by nodes
with rank $i$ and $i+1$ in $V$. Assume that the node pointed by node with rank $i$ is in hypernode $V'$. This means that
we have narrowed down the search in a subproblem consisting of $2^{i+1}-2^i=2^i$ nodes after having made $\ell-i$ steps
due to the backward exponential search. The procedure is applied again from node with rank $i$ in $V'$ until we find
node $u$. This node in $V'$ can be located in $O(1)$ steps, since the node with rank $i$ of hypernode $V$ maintains a link
to node with rank $i$ in hypernode $V'$. Thus, the number of steps is at most $\ell$. The vertical search
on a path from a node towards the root or a leaf is exactly the same as before.
\qed
\end{proof}

A range query $[a,b]$ reports all elements $x$ such that $x \in [a,b]$. A range query $[a,b]$ initiated at node $v$,
invokes a search operation for element $a$. Node $u$ that contains $a$ returns to $v$ all elements in this range.
If all elements of $u$ are reported then the range query is forwarded to the right adjacent node (inorder traversal)
and continues until an element larger than $b$ is reached for the first time.

\subsubsection{Updates and Load Balancing.}
Assume that an update operation is initiated at node $v$ involving element $\alpha$. By invoking a search operation
we locate node $u$ with range containing element $\alpha$. Finally, the update operation is performed on $u$.
The main issue is how to balance the load to all nodes of the overlay as much equally as possible. To do that we
employ the machinery developed in Section~\ref{sec:load}. Assume that $w$ is the node for which the redistribution
must be applied. It remains to determine how the redistribution will be realized. An implementation of this redistribution follows.

First we make a scan of all nodes in the subtree of $w$ by forwarding a message which simply counts the number of nodes and the number of elements in the subtree. Finally, this message ends up in the leftmost leaf of the subtree of $w$. Thus, $w$ now knows exactly how many elements should be distributed in each node in order to have a uniform load. Then, a data migration procedure is initiated.

 The idea is to migrate the elements to their final destination nodes in a simple step and in an inorder traversal fashion which is facilitated by adjacency links. The link \textit{dest} facilitates the transfer of elements between nodes and at the same time functions as a token which designates the node that implements at the moment the data migration. Starting from the rightmost node of the rightmost bucket in the subtree of $w$, it checks whether the number of elements is less or more than the ideal load. If they are less, then by using the \emph{dest} link the necessary number of elements is transferred from the designated node to the node containing \emph{dest}. If they are more, the necessary number of elements are moved to the node designated by \emph{dest}. If during this procedure the node designated by \emph{dest} fills up (meaning it reaches the desired load) or empties (meaning we transferred a lot of elements) then \emph{dest} is moved to the next node w.r.t. the inorder traversal. When the node containing \emph{dest} has reached its ideal load then \emph{dest} is moved to the next node w.r.t. the inorder traversal and the procedure continues. This procedure requires a linear number of messages w.r.t. the number of elements in the subtree of node $w$.

The cost for the redistribution of a node $v$ is $O(|v|\log{N})$ for the case of $O(\log{N})$ space per node
or $O(|v|)$ for the case of $O(1)$ space per node. This is because, during the transfer of elements the routing
tables must be reconstructed. The following lemma states that the load balancing is efficient in an amortized
sense when the structure is subject to insertions and deletions of elements.
\begin{lemma} \label{lem:complrebal}
 The load rebalancing operation of the index has an amortized cost of $O(\log{N})$.
\end{lemma}

\begin{proof}
This is a direct implication of Theorem~\ref{thm:redistcost} and the space used by the nodes.
\qed
\end{proof}

One final comment is that the redistribution of elements may be affected by the redistribution of nodes in the
weight-balanced overlay. In order to avoid such a phenomenon, the redistribution of nodes in the subtree of node
$v$ in the overlay is preceded by a redistribution of elements.

\subsection{Other Efficiency Issues and the Main Result}

We are now ready to tackle the congestion and the fault-tolerance of the $D^2$-tree overlay,
and to present the main results of this work.

\subsubsection{Congestion.}
We assume that a sequence of
searches $s_1,s_2,\ldots,s_N$ is initiated from each of the $N$ nodes of the overlay. We assume
that search $s_i$ is looking for an element residing in a node $z_i$ (target node for $s_i$).
The target nodes $z_1,z_2,\ldots,z_N$ are chosen independently and uniformly at random from all
nodes of the overlay. There are two phases in the search. The first is the horizontal search phase,
which makes use of the routing tables, and the second is the vertical search phase on a path from a
node either towards the root or towards a leaf.

To establish a bound on the congestion, we need to provide bounds on the horizontal and
vertical searches. These bounds are provided by Lemmata~\ref{lem:hsearch} and \ref{lem:vsearch}
below. Before proving these lemmata, we need the following result.

\begin{lemma} \label{lem:destinations}
 The number of searches that stop at a node $v$ at level $\ell$ during the horizontal phase
 of the search is $O(1)$ in expectation.
\end{lemma}
\begin{proof}
Since the destinations are chosen uniformly at random, the destination nodes at level $\ell$ for searches starting from this level depend on the weight of each node plus the weight of the nodes on the path to the root which is almost equal. The weight of each node at level $\ell$ is approximately equal for all nodes. Thus, it is expected that $O(1)$ searches will have as a destination any node at level $\ell$.
\qed
\end{proof}

The following lemma bounds the congestion due to the horizontal search.

\begin{lemma} \label{lem:hsearch}
 The horizontal phase of the search at level $\ell$ contributes to congestion $O(\ell)$ in expectation
 at each node of this level.
\end{lemma}
\begin{proof}
Level $\ell$ contains $O(2^\ell)$ nodes. We number the nodes from left to right by $0,1,\ldots$. A {\em path} from a node $j$ to a node $k$ is the sequence of nodes that we access when we search from node $j$ to find node $k$ at level $\ell$ by using the routing tables. Let $X_{i,j}$ be the random indicator variable defined as follows:
\[
X_{i,j}=\left\{ \begin{array}{c l}
				1 & \mbox{ if node $i$ is in the path that starts from node $j$} \\
				0 & \mbox{ otherwise}
				\end{array}\right.%\}
\]
$X_{i,j}$ is a random variable since node $j$ can choose its target among all nodes at level $\ell$ uniformly at random as implied by Lemma~\ref{lem:destinations}. The following quantity bounds the expected number of paths passing through an arbitrary node $i$ when all searches from nodes at level $\ell$ are accounted for.
\[E\left[ \sum_{j=0}^{O(2^{\ell})}{X_{i,j}}\right] = \sum_{j=0}^{O(2^{\ell})}{E\left[X_{i,j}\right]} \]
%
% which by linearity of expectation becomes
% \[\sum_{j=0}^{O(2^{\ell})}{E\left[X_{i,j}\right]}\]
%
Since $X_{i,j}$ is a random indicator variable it follows that
\[E\left[X_{i,j}\right] = \Pr\left\{X_{i,j}=1\right\}\]
This probability is equal to the number of paths going through $i$ divided by the total number of paths starting from $j$ and ending at all nodes of level $\ell$.
\[\Pr\left\{X_{i,j}=1\right\}=\frac{\mbox{\# of paths passing through $i$ from $j$}}{\mbox{Total number of paths starting at $j$}}\]
The total number of paths starting from $j$ to all nodes of level $\ell$ is equal to the number of target nodes which is $O(2^{\ell})$. Note that we only count the number of search paths as defined by the search procedure between two nodes and not all possible paths.

It is a little trickier to compute the number of paths going through node $i$.
The crucial observation is that the binary representations of the nodes, in their
left to right numbering at level $\ell$, provide
a way to count the number of paths passing through a particular node. Let the
binary representation of node $i$ be $i_{\ell-1}\ldots i_1 i_0$, where $i_{\ell-1}$
is the most significant bit. Then, if there is a link of length $2^{\ell-1}$
between node $i$ and node $j$ it holds that $i_{\ell-2}\ldots i_1 i_0 = j_{\ell-2}\ldots j_1 j_0$. The following observation holds.

\begin{observation}\label{obs:access}
Node $i$ will be accessed by a link of length at most $2^{m}$ in a search path starting from $j$
if $i_{m-1}\ldots i_1 i_0  = j_{m-1}\ldots j_1 j_0$.
\end{observation}
\begin{proof}
This is an implication of the construction of the routing tables as well as from the fact that during searching the sequence of links that are followed are of monotonically decreasing length by powers of $2$.
%Assume that the predecessor of $i$ in the path from $j$ is $i'$. If $i'=j$ then the observation is a straightforward implication by the fact that the lengths of the links are powers of $2$. Otherwise, assume that
\qed
\end{proof}

Thus, we have to compute:
\[\frac{1}{O(2^{\ell})}\sum_{j=0}^{O(2^{\ell})}{\left(\mbox{\# of paths passing through $i$ from $j$}\right)} \]
The number of paths that go through $i$ starting from $j$ with destination any node at level $\ell$ can be deduced by using Observation~\ref{obs:access} and the properties of the binary representations. In particular, if the $m$ less significant bits of numbers $i$ and $j$ are equal and $i_m \neq j_m$, then at most $2^m$ paths go through $i$ by Observation~\ref{obs:access}. The number of different nodes $j$ that go through $i$ in this case is $2^{\ell-m}$ since those are the possible numbers that have the $m$ least significant bits the same as $i$. Thus, the previous sum can be expressed by summing over all possible $m$:
\[\frac{1}{O(2^{\ell})} \sum_{m=0}^{\ell-1}{2^{\ell-m}2^m}=\frac{\ell 2^{\ell}}{O(2^{\ell})}=O(\ell)\]
and the lemma follows.
\qed
\end{proof}

The following lemma bounds the congestion due to the vertical search.

\begin{lemma} \label{lem:vsearch}
 The vertical phase of the search starting at level $\ell$ contributes to
 congestion $O(1)$ in expectation at each node in its subtree or on the path to the root.
\end{lemma}
\begin{proof}
By Lemma~\ref{lem:destinations}, only an expected $O(1)$ number of searches will stop at any node due to the horizontal search phase. Assume a node $u$ at level $\ell$. This node has $\ell-1$ ancestors and $2^{H-\ell}$ descendants. Thus, in total at most $O(2^{H-\ell}+\ell)$ searches in expectation can affect node $u$. We start by investigating how ancestors affect node $u$. The ancestor at level $\ell-1$ can choose between two children, the one of which is $u$, as well as from its path of ancestors. Thus, the probability of choosing $u$ is $O\left(\frac{2^{H-\ell}}{2^{H-\ell+1}+\ell-2}\right)$. In general, the probability of node $z$ at level $\ell'<\ell$ going through $u$ is $O\left(\frac{2^{H-\ell}}{2^{H-\ell'}+\ell'-1}\right)$. Thus, the expected number of searches going through $u$ due to its ancestors is
\begin{equation}
\label{eq:v-anc}
 \sum_{\ell'=1}^{\ell}{O\left(\frac{2^{H-\ell}}{2^{H-\ell'}+\ell'-1}\right)} = O\left(2^{H-\ell}\sum_{\ell'=1}^{\ell}{\frac{1}{2^{H-\ell'}}}\right) = O(1)
\end{equation}
 Now we move to the descendants of $u$. The probability that the leaves of the subtree of $u$ go through $u$ during a search is $O\left(\ell\frac{2^{H-\ell}}{n}\right)$. This is because the probability of choosing any node as a destination node of the search operation is $\frac{1}{n}$, the number of leaves is $O(2^{H-\ell})$ and there are $\ell$ nodes in total from $u$ to the root. Similarly, for the $i$-th level, $i>\ell$, the probability of going through $u$ is $O\left(\ell\frac{2^{H-\ell-i}}{n} \right)$. Thus, in total we get that the expected number of searches going through $u$ from its descendants is
\begin{equation}
\label{eq:v-desc}
 \sum_{i=0}^{\ell-1}{O\left(\ell\frac{2^{H-\ell-i}}{n} \right)} = O\left(\frac{\ell2^{H-\ell+1}}{n}\right)= O\left(\frac{\ell2^{H-\ell+1}}{2^{H-1}}\right)= O\left(\frac{\ell}{2^{\ell-2}}\right) = O(1)
\end{equation}
By (\ref{eq:v-anc}) and (\ref{eq:v-desc}) we get the lemma.
 \qed
\end{proof}

The following theorem establishes the congestion bound.

\begin{theorem}
\label{thm:congestion}
 The (expected) congestion due to the search operations is $O\left(\frac{\log{N}}{N}\right)$
 in a $D^2$-tree with $N$ nodes, when each node uses $O(\log{N})$ space.
\end{theorem}

\begin{proof}
By Lemmata~\ref{lem:hsearch} and \ref{lem:vsearch}, we deduce that $O(\log{N})$ searches in
expectation will go through each node of the tree. Since the tree has $N$ nodes, the
theorem is established.
\qed
\end{proof}

The following theorem extends Theorem~\ref{thm:congestion} by using $O(1)$ space per node.

\begin{theorem} \label{thm:congestion_space}
The (expected) congestion due to the search operations is $O(\frac{\log{N}}{N})$
in a $D^2$-tree with $N$ nodes, where each node uses $O(1)$ space.
\end{theorem}

\begin{proof}
This proof is very similar to the proof of Theorem~\ref{thm:congestion}
and we simply sketch it. Lemma~\ref{lem:destinations} still holds.
Searching is again divided into two phases. The vertical search phase
is identical to the one in %%% Theorem~\ref{thm:congestion} and as a result Lemma~\ref{lem:vsearch} holds.
Lemma~\ref{lem:vsearch} and hence this lemma still holds.
However, horizontal search has slightly changed and Observation~\ref{obs:access} is not valid anymore.
First, the search always starts from the highest rank node in a hypernode $V$ which
results in $O(\ell)$ accesses from the searches that start from all nodes of $V$. From this point and on, the horizontal search is similar to the
one of %%% Theorem~\ref{thm:congestion}.
Lemma~\ref{lem:hsearch}. The proof that the congestion remains optimal is a result of the following
similar argument to Lemma~\ref{lem:hsearch}. The probability that a node $i$ will be part of the search path which
starts at node $j$ is large for very few nodes $j$ among the $O(2^{\ell})$ such possible nodes. Most of the nodes
have a very small probability of using node $i$, since node $i$ can be accessed after $O(\ell)$ steps. This follows
directly from Lemma~\ref{lem:destinations}. Using this fact and the fact that highest rank nodes have
at least $O(\ell)$ accesses we are driven to the conclusion that the expected bound on the number of accesses
to nodes of level $\ell$ due to the horizontal search is $O(1)$ and the theorem follows.
\qed
\end{proof}

\subsubsection{Fault Tolerance.}
If a node $v$ discovers (during the execution of an operation) that node $u$ is unreachable, then it contacts a
sibling of $u$ through the routing tables of the siblings of $v$ (by making use of Lemma \ref{lem:properties}(ii)).
This sibling of $u$ is able by Lemma~\ref{lem:properties}(ii) (or Lemma~\ref{lem:hyper_routing_con_sibling}) to
reconstruct all links of node $u$ and a node departure for $u$ is initiated, which resolves this failure.

Searches and updates in the $D^2$-tree do not tend to favour any node, and in particular nodes near the root.
This is a direct consequence of the way the search operation is implemented by first moving horizontally
at the same level as the node that initiated the search and then by moving vertically (see Theorem~\ref{thm:congestion_space}).
As a result, near to root nodes are not crucial and their failure will not cause more problems than the failure of any node.
However, a single node can be easily disconnected from the overlay simply when all nodes with which it is connected fail.
This means that 4 failures (two adjacent nodes and two children) are enough to disconnect the root (recall that the routing
table of the root is empty). For the $O(1)$ space per node solution, a $O(1)$ number of failures is enough to disconnect any node.
For the $O(\log{N})$ space per node solution, a node at level $\ell$ can be disconnected after $O(\ell)$ failures in the worst-case.

When routing tables have $O(\log{N})$ size, to disconnect a group of $k$ nodes at least $k$ failures must happen.
The most easily disconnected nodes are those which are near the root since their routing tables are small in size.
Thus, they can be disconnected by simply letting their respective adjacent nodes (which are leaves) fail which
provides the bound. When routing tables have $O(1)$ size, fault tolerance is naturally deteriorated. When
the representative of a bucket fails then the leftmost node among the nodes of the bucket replaces it, initiating
a departure operation.

\subsubsection{Main Result.}
We are now ready for the main result of this work.

\begin{theorem} \label{thm:bb_space}
A $D^2$-tree overlay with $N$ nodes and $n$ data elements residing on them achieves:
(i) $O(1)$ space per node;
(ii) deterministic $O(\log N)$ searching cost;
(iii) deterministic amortized $O(\log{N})$ update cost both for element update and for node joins and departures;
(iv) \emph{optimal} congestion of $O\left(\frac{\log{N}}{N}\right)$ expected cost;
(v) deterministic amortized $O(\log{n})$ bound for load-balancing.
The $D^2$-tree overlay supports ordered data queries optimally, and tolerates
node failures.
\end{theorem}
\begin{proof}
Space usage is $O(1)$ by construction. The search cost follows from Lemma~\ref{lem:search}.
Node join and departures are $O(\log{N})$ amortized by Lemma~\ref{lem:migrate} and the fact
that $f(n)=O(N)$. The congestion bound comes from Theorem~\ref{thm:congestion_space}.
Finally, the load-balancing bound comes from Lemma~\ref{lem:complrebal}.
\qed
\end{proof}

\REMOVED{
\subsection{Applications}

The balancing scheme can be applied straightforwardly to BATON \cite{jov05}. BATON is a balanced tree-like overlay
that satisfies the specifications set in the Introduction. The same goes also for Skip Graphs \cite{as03} with
the exception that the specifications hold probabilistically and thus the bounds are also probabilistic.

We provide a technique that lazily updates the weights on the nodes of a tree. This technique is interesting by
itself and can be straightforwardly applied to weighted balanced trees \cite{av03} in the Pointer Machine model
of computation for single processor internal memory machines. In this manner, the update of balancing information
is supported in $O(1)$ amortized time instead of $O(\log{n})$ as it was known until now.
}

\section{Discussion and Future Work}
\label{sec:conclusion}

Our load-balancing scheme (Section~\ref{sec:load}) can be applied straightforwardly to BATON \cite{jov05}. BATON is a balanced tree-like overlay
that satisfies the specifications set in the Introduction. The same goes also for Skip Graphs \cite{as03} with
the exception that the specifications hold probabilistically and thus the bounds are also probabilistic. Additionally, it
provides a mechanism to control the bucket size of \cite{akk04}.

We provide a technique that lazily updates the weights on the nodes of a tree (Section~\ref{ssec:rebalance}). This technique is interesting by
itself and can be straightforwardly applied to weighted balanced trees \cite{av03} in the Pointer Machine model
of computation for single processor internal memory machines. In this manner, the update of balancing information
is supported in $O(1)$ amortized time, an improvement over the currently best known bound of $O(\log n)$.

Future work includes the extension of the load-balancing mechanism to accommodate weighted
elements (weights representing preference). Additionally, the load balancing mechanism provides amortized complexities which results in the existence of very few indeed but very costly rebalancing operations (imagine the root being redistributed). To fully tackle the existence of churn, one needs to come up with worst-case complexities for the load balancing mechanism. Note that \textit{churn} is the collective effect created by independent burstly arrivals and departures of nodes.

With respect to the overlay, future work includes tackling multidimensional data, integrating the network topology with the overlay topology as well as taking into account locality of reference.

It is also an open problem the application of the proposed balancing scheme to the BATON$^*$ \cite{JOTVZ06}
structure (the latest version of BATON), where the overlay structure is a tree with height $O(\log_{m}{N})$ with each node having $O(m)$ children.

Finally, the mechanisms we provide require extensive experimental verification.

\end{document}